\documentclass[11pt]{article} 
\usepackage{pslatex}

\usepackage{amsfonts}
\usepackage{amsmath}
\usepackage{fullpage}

\newtheorem{theorem}{Theorem}
\newtheorem{lemma}[theorem]{Lemma}
\newtheorem{definition}{Definition}

\newcommand{\qed}{\hfill\rule{7pt}{7pt}}
\newenvironment{proof}{\noindent{\bf Proof:}}{\qed\medskip}
\newcommand{\defeq}{\stackrel{def}{=}}

\newcommand{\etal} {{\it et al.}}

\newcommand{\nb}{\mbox{nb}}
\newcommand{\Svec}{\mathbf{S}}
\newcommand{\pvec}{\mathbf{p}}
\newcommand{\qvec}{\mathbf{q}}
\newcommand{\qbase}{\overline{\mathbf{q}}}
\newcommand{\eps}{\pmb{\epsilon}}
\newcommand{\rvec}{\mathbf{r}}
\newcommand{\phat}{\hat{\pvec}}
\newcommand{\qhat}{\hat{\qvec}}
\newcommand{\pert}{\mathbf{t}} 
\newcommand{\Dvec}{\mathbf{D}}
\newcommand{\lambar}{\overline{\lambda}}

\begin{document}
\title{Multi-outcome and Multidimensional Market Scoring Rules\\
        (Manuscript)
      }
\author{
Lance Fortnow%
\thanks{Supported in part by NSF grant CCF-0829754}
       \\Northwestern University
       \\{\it fortnow@eecs.northwestern.edu}
\and
Rahul Sami%
\thanks{Supported in part by NSF grant CCF-0812042}
       \\University of Michigan
       \\{\it rsami@umich.edu}
}
\maketitle
\begin{abstract}
Hanson's market scoring rules allow us to design a prediction
market that still gives useful information even if we have an
illiquid market with a limited number of budget-constrained
agents. Each agent can ``move'' the current price of a market
towards their prediction.

While this movement still occurs in multi-outcome or
multidimensional markets we show that no market-scoring rule,
under reasonable conditions, always moves the price directly towards
beliefs of the agents. We present a modified version of a market scoring rule for budget-limited traders, and show that it does have the property that, from any starting position, optimal trade by a budget-limited trader will result in the market being moved towards the trader's true belief. This mechanism also retains several attractive strategic properties of the market scoring rule.

\end{abstract}

\section{Introduction}\label{sec:intro}
Prediction markets are markets set up for the primary purpose of aggregating information to forecast future events.
For a heavily traded prediction market, such as whether Obama
will be reelected in 2012, we can get a very precise prediction
from the market. As of January 30, 2012, the Intrade bid-ask
spread gives a prediction between 53.3\% and 53.9\% of
reelection. For other less liquid markets, the bid-ask spread gives a less
informative prediction. The Intrade bid-ask spread for whether
Sarah Palin will endorse Mitt Romney is sitting between 15.1\%
and 99.0\%.

To recover information in thin, illiquid markets, Robin Hanson
developed the market scoring rule~\cite{Hanson-MSR}. With this form prediction market, a market maker always quotes a single current price for each future events. If the agents have a belief different than the current price, a
trade will give them an a priori expected positive profit. The market maker adjusts the price as trades are made. If traders are risk-neutral and do not face budget constraints, Hanson showed that it is optimal for a trader to trade until the price matches her beliefs; the final market prices can thus be interpreted as updated probabilities for the forecast events.
The market maker loses money in expectation, but the total loss of the market maker can be bounded.

The further the change in price, the higher the up-front cost, and worst-case loss, that the trader has to bear. In real-world prediction markets, traders often have budget constraints that limit how much loss they can bear. In fact, many real money prediction market platforms impose limits on how much money an agent can invest over the entire platform, and most play-money markets limit the ability of traders to convert external money to virtual points that can be invested. Thus, limited budgets are an inherent feature of the context in which most prediction markets are run.
If an agent is budget-limited she can still trade, but may not be able to move the prices to match her beliefs.
For two-outcome markets, like the ones described above, the
market scoring rule induces straightforward behavior in budget-limited agents:
An agent will move the prices of the securities directly towards her belief, and the final probability forecast will therefore be a mixture of the original market probabilities and the agent's true belief. In this paper, we explore the behavior of budget-limited traders in markets with multiple outcomes or multi-dimensional outcomes.

Many natural events for which we seek forecasts have multiple outcomes.
Consider the following
scenarios:
\begin{itemize}
\item A market over Republican candidates for the U.S. Presidential nomination:
    For this market, a forecast can be represented by a four-tuple probability vector
    $(p_R,p_G,p_P,p_S)$ for the probabilities that Romney,
    Gingrich, Paul and Santorum will be nominated.
\item A market predicting the latitude and longitude of the
    next tornado strike in Kansas. Here, the forecasts consist of probability distributions over a two-dimensional grid.
\item A market predicting the likelihood of hurricane strikes for
    each year over the next three years. Here, a forecast consists of a probability distribution over $\{0,1\}^3$; however, it may be natural to assume independence between different years.
\end{itemize}

For each of these markets, we can set up prediction markets based on proper scoring rules.
One way to do so is to break the markets into their components, a separate market for each candidate, for latitude
and longitude and for each year. More generally, however, we may have a complex combinatorial market whose payoffs are dependent on the combination of events that occur~\cite{Hanson-MSR}. In either case, a risk-neutral trader without budget constraints would optimize her payoffs by moving the market forecast to match her true belief. When traders are budget-limited, the optimal behavior may depend on the particular proper scoring rule used. One natural goal is to ask:
Can we find a proper scoring rule such that, for any positive budget an agent has, her optimal trade is always to move directly towards her true belief, so that the updated forecast is a mixture of the original forecast distribution and her true belief?

 Our main results show that this is impossible. In section~\ref{sec:3outcome} we show that, under mild smoothness assumptions, no scoring rule over three or more outcomes can have this property. For example, in the Republican candidate setting, we cannot find any market scoring rule that will induce a budget-constrained trader to move the market forecast in a straight line towards her true belief. 
 Further, we show that there are neighborhoods of possible beliefs that all result in exactly the same report; this shows that sequential aggregation in the MSR with budget constraints can be imperfect, because even if all budgets are public information, a trader's belief cannot be inferred from her report.
 The results in section~\ref{sec:3outcome} relies on the space of beliefs being unrestricted, and thus seems to leave open the possibility of a positive result when there are natural restrictions, such as independence of hurricane probabilities in different years. However, we show that this is not the case in section~\ref{sec:2dim}: Even when the outcome space consists of two binary events that are assumed to be independent, no smooth market scoring rule can induce a budget-limited agent to move the market distribution directly towards her true belief. Importantly, this latter result holds without assuming that the scoring rule is itself simple or separable, i.e., it may be a combinatorial scoring rule.

 In section~\ref{sec:mechanism}, we present a modified market scoring rule, the
Scaled Scoring Mechanism, that has this property of inducing budget-limited agents to move the forecast directly towards their true belief. This mechanism requires a trader to report her budget as well as her belief; we assume that the budget must be deposited up-front, and hence cannot be exaggerated (although it may be under-reported). We prove that a trader's optimal trade is to report her true budget and belief, and hence, the updated market forecast is always a mixture of her true belief and the original forecast. Further, we show that this mechanism retains many of the desirable properties of the market scoring rule: The market-maker's loss is bounded, and the traders' profit along a sequence of updates is sub-additive, so that a trader cannot gain a profit from splitting up her trade into two or more successive trades by different identities.
As this mechanism induces full information revelation from budget-limited traders, it also enables effective sequential aggregation of information: A trader who has seen previous traders' trades can infer their beliefs; hence, in a Bayesian world, a trader can generically condition her beliefs on all previous traders' private information.

 The rest of this paper is structured as follows: We discuss related work in section~\ref{sec:related}. We introduce our formal model in section~\ref{sec:model}. We prove the impossibility results in section~\ref{sec:3outcome} and~\ref{sec:2dim}, and introduce the new mechanism in section~\ref{sec:mechanism}. Finally, in section~\ref{sec:discussion}, we discuss extensions and limitations of our results, and directions for future work.

\subsection{Related Work}\label{sec:related}
 There is a long literature on the use of proper scoring rules to incentivize
accurate probabilistic forecasts of events~\cite{Good52}~\cite{WinklerMurphy68}.
Proper scoring rules are characterized by the property that, given any belief
$\pvec$, the report that globally maximizes the forecaster's expected (under
belief $\pvec$) reward is the honest report $\qvec = \pvec$. Subsequent
research on {\em effective scoring rules} studied the design of scoring rules
with a stronger monotonicity property. A scoring rule is said to be
effective with respect to a given metric $L$ over the space of probability
distributions if, given any two possible reports $\qvec, \qvec'$ such that
$\qvec$ is closer (under metric $L$) to the true belief $\pvec'$ than $\qvec'$,
the expected score of reporting $\qvec$ must be greater than the expected score
of reporting $\qvec'$~\cite{Friedman83}. Friedman~\cite{Friedman83} has shown that
there are scoring rules that are effective with respect to some metrics,
including the $L_2$ metric. Our negative results present a striking contrast:
We show that no scoring rule can induce a budget-limited agent to move
only along the direction of steepest decrease in $L_2$ distance to the
true belief.

  Hanson's market scoring rules~\cite{Hanson-MSR} transform scoring rules into a mechanism for obtaining sequential forecasts from multiple forecasters. Hanson showed that this sequential scoring form can serve as an automated market maker for prediction markets.

  Budget limits can be viewed as a special form of risk aversion.
Allen~\cite{Allen-scoring} described a lottery technique to extend scoring rules to
forecasters with arbitrary and unknown risk aversion. Dimitrov et al.~\cite{DSE09} show
that this technique can be extended to a sequential forecasting mechanism. However, the
mechanism of Dimitrov et al. has the undesirable property that expected profits shrink exponentially as the number of traders grows, and moreover, they prove that this is unavoidable in the context of mechanisms for arbitrarily risk averse agents. The mechanism we present here does not have this undesirable property. 

  Manski~\cite{Manski} and Wolfers and Zitzewitz~\cite{WZ-budget} have also studied the effect of budget limits on the aggregative properties of prediction markets. However, our model differs from their models in significant ways: 
They study a stylized unsubsidized market model in which clearing prices are determined in a one-shot game, whereas we model a market scoring rule with common0knowledge Bayesians who update beliefs after each trade.

\section{Model}~\label{sec:model}
\paragraph*{Outcomes and Distributions:}
Consider an event with more than two possible outcomes. In
order to simplify the notation, we begin with a three-outcome
event, and let the outcomes be denoted ${X,Y,Z}$. An agent has
a belief $\pvec=(p_X, p_Y, p_Z)$ on the probabilities of the
various outcomes. When asked for a forecast, the agent may
report a forecast $\qvec=(q_X, q_Y, q_Z)$. In general, for
unrestricted beliefs over a $k$-outcome event, $\pvec$ and
$\qvec$ can be represented as $k$-dimensional vectors, with the
implicit constraints that the elements are non-negative and sum
to $1$.

 We also consider a special, natural, class of {\em restricted} distributions, which
arise when the event being forecast itself has a dimensional
structure. For clarity, we defer the description of this
restricted class to section~\ref{sec:2dim}.

\paragraph*{Scoring Rules}
The agent is rewarded based on a scoring rule $\Svec$. We model
$\Svec$ as a mapping from the space of forecasts $\qvec$ to the
space of payments $\Svec(\qvec) \equiv (S_X(\qvec), S_Y(\qvec),
S_Z(\qvec))$  that will be made for each outcome. For
unrestricted beliefs over three-outcome events, the space of
all probability distributions  $\Delta =\Delta(\{X,Y,Z\})$ is a
two-dimensional space. The scoring rule $\Svec$ can be viewed
as a two-dimensional surface in $\Re^{3}$; each point
$\Svec(\qvec) \in \Re^3$ corresponds to a particular value of
$\qvec$.

\noindent{\bf Assumptions:} We assume that the scoring rule
$\Svec(\qvec)$ is continuous and differentiable over the
interior of the space of possible reports $\qvec$. (Actually,
all we really need is that it is continuous and differentiable
over some open set of reports.)
We also assume that each component of the scoring rule is
quasiconcave: $\forall \pvec, \qvec \; S_X(0.5 \pvec +
0.5\qvec) \geq \min (S_X(\pvec), S_X(\qvec)$, with equality
only if $S_X(\pvec) = S_X(\qvec)$; and likewise for $S_Y, S_Z$.
This is a mild condition (weaker than monotonicity of the score
functions), that is true for all commonly used scoring rules.
It is used to argue that, when moving from $\qvec$ to $\pvec$,
the budget increases monotonically along the path, and so any
point on the path is optimal for some budget. We assume that
agents are risk neutral (up to their hard budget constraints,
described below); thus, an agent's goal is to maximize her
expected score $\pvec \cdot \Svec(\qvec)$.

\paragraph*{Market Scoring Rules and Sequential Updates}
 Hanson's market scoring rules (MSR)~\cite{Hanson-MSR} are a form of prediction market that
extends scoring rules to a sequential information aggregation
mechanism. Consider a set of agents $\{1, 2, \ldots, m\}$
making sequential forecasts $\qvec_1, \qvec_2, \ldots,
\qvec_m$. With the market scoring rule based on scoring rule
$\Svec$, each agent $i$ receives a net payoff equal to the
difference between the score of $\qvec_i$ and $\qvec_{i-1}$.
Thus, the net expected payoff of an agent $i$ who believes a
distribution $\pvec$ is:
\[
  \pvec.[\Svec(\qvec_i) - \Svec(\qvec_{i-1})]
\]
Here, $\qvec_0$ is an initial default distribution specified by
the mechanism.

 Hanson observed that, with the MSR mechanism, if each agent trades just once and there are no budget constraints, the optimal action for each trader is to report $\qvec_i$ that
matches her true belief. Consequently, future traders who can
observe the past reports (prices) in the market can condition
on all previous agents' beliefs about the event. Generically,
this leads to perfect aggregation of the $m$ agents' combined
information.

  In our analysis of scoring with budget constraints, each agent's report $\qvec$ is scored with respect to a {\em reference point} $\qbase$, and thus her score is given by
the difference $\Svec(\qvec) - \Svec(\qbase)$. The notion of a
reference point or default $\qbase$ is an essential ingredient
of a model with budget constraints, because it describes what
the forecast will be when an agent has $0$ budget. The
reference point may be the previous traders' report
$\qvec_{i-i}$, as in the market scoring rule, but it may also
be determined differently; in section~\ref{sec:mechanism} we
present a mechanism with a different way of determining a
reference point for each trader.

\paragraph*{Budgets and Constrained Effectiveness}

 We are interested in scoring rules that are strictly proper: The unique
optimal $\qvec$, given belief $\pvec$, is $\qvec \equiv \pvec$.
However, we are further interested in a {\em budget-limited}
version of this property.  Informally, we want the scoring rule
to have the property that, given any initial value $\qbase$,
and any specified budget $b$,the optimal report within the
budget feasible set is a report $\qvec$ along the path from
$\qbase$ to $\pvec$.

 In order to make this more rigorous, we need to define the notion of a budget
more formally.
 One natural definition of the budget required to move from
$\qbase$ to $\qvec$ is in terms of the maximum loss that the
agent can incur from such a move; this is the amount of money
that an agent would require to have to avoid defaulting on the
mechanism.
\begin{definition}
 The {\bf natural budget} $\nb(\qbase,\qvec)$ required to move from $\qbase$ to $\qvec$,
given a scoring rule $\Svec$, is defined by
\[
   \nb(\qbase, \qvec) \defeq \max_{\pvec} \left[ \pvec \cdot (\Svec(\qbase) - \Svec(\qvec)) \right] = \max \{ S_X(\qbase) - S_X(\qvec), S_Y(\qbase) - S_Y(\qvec), S_Z(\qbase) - S_Z(\qvec) \}
\]
\end{definition}

 Together, the reference point $\qbase$ and an initial budget holding $b$ determine the
range of possible reports that an agent can make while avoiding
default under any outcome; the natural budget constraint
determines the set of feasible reports. Thus, the natural
budget constraint determines a budget-constrained agents'
choice set under the market scoring rule as typically
implemented. However, it may be possible for an alternative
mechanism to further restrict the set of allowable reports by
an agent; in section~\ref{sec:mechanism}, we will exploit this
to obtain a mechanism with better aggregative properties.


 Now, the {\em natural budget-constrained optimal report $\qvec^*$}, with budget $b$, is
the choice among all $\qvec$ such that $\nb(\qbase, \qvec) \leq
b$, that maximizes $\pvec \cdot \Svec(\qvec)$. Let
$\qvec^*(\pvec, \qbase, b)$ denote the optimal report with
budget $b$.

We can now state our desired property:
\begin{definition}
 A scoring rule $\Svec$ satisfies the budget-constrained truthfulness property
if the budget-constrained optimal choice is always a mixture of
the initial distribution $\qbase$ and the belief $\pvec$, and
as close to $\pvec$ as possible:
\[
\forall \pvec, \qbase, b > 0 \;\;  \qvec^*(\pvec, \qbase, b) = \alpha\pvec + (1-\alpha) \qbase
\]
where
\[
  \alpha \in [0,1] = \max \{ \alpha| b(\qbase, \alpha\pvec + (1-\alpha) \qbase) \leq b \}
\]
\end{definition}

This definition inherently implies that $\Svec$ is a proper
scoring rule (by taking a sufficiently large $b$); however, it
is significantly stronger than the properness condition.

\section{An impossibility result for unrestricted multi-outcome distributions}\label{sec:3outcome}
For two-outcome events, the budget-constrained truthfulness
property holds trivially for all common scoring rules, as the
natural budget and expected score both increase monotonically
as $\qvec$ moves from $\qbase$ towards the agent's true belief
$\pvec$. However, well-known scoring rules do not satisfy this
for general multiple-outcome events. Here, we show that this is
not accidental: We will prove that no scoring rule (satisfying
our technical assumptions) over three or more outcomes can
satisfy the budget-constrained truthfulness property. We will
assume without loss of generality that there are exactly three
outcomes $X, Y, Z$; the result holds {\em a fortiori} for
$k>3$. 

We will prove the impossibility result by contradiction. Suppose that we had
a scoring rule $\Svec$ such that $\Svec$ satisfied the
budget-constrained truthfulness property.

We first make a simple observation about the local structure of
the scoring surface:
\begin{lemma}\label{lemma:local}
  Consider a point $\pvec$ in the interior of $\Delta$, and consider a
neighboring point of the form $\pvec + \delta \eps$, where
$\eps$ is a nonzero vector such that $\epsilon_X + \epsilon_Y +
\epsilon_Z = 0$, and $\delta$ is an infinitesimally small
quantity.  Then, there exists a direction $\eps$ such that
$\frac{dS_X(\pvec + \delta \eps)}{d\delta} =0$.  Further, for
this $\eps$, $\frac{dS_Y(\pvec + \delta \eps)}{d\delta}$ and
$\frac{dS_Z(\pvec + \delta \eps)}{d\delta}$ must be nonzero and
have opposite signs.
\end{lemma}
\begin{proof}
 Expanding $S_X$ around $\pvec$ in terms of its partial derivatives, we have:
\[
  \frac{dS_X(\pvec+ \delta \eps)}{d\delta} = \epsilon_X \frac{\partial S_X}{\partial p_X} +  \epsilon_Y \frac{\partial S_X}{\partial p_Y} - (\epsilon_X + \epsilon_Y) \frac{\partial S_X}{\partial p_Z}
\]
Setting the LHS to $0$, we get one linear constraint on
$\eps_X$ and $\eps_Y$; any solution to this will meet our
purposes.

For the second part of the statement, consider the point $\qvec
= \pvec+ \eps d\delta$. We have $S_X(\qvec) = S_X(\pvec)$.
However, as this is a proper scoring rule, we must have
$\pvec\cdot[S(\pvec) -S(\qvec)] >0$ and $\qvec\cdot[S(\pvec)
-S(\qvec)] <0$.  As $\pvec$ and $\qvec$ have non-negative
entries, the only way in which this can happen is if $S(\pvec)
- S(\qvec)$ has one positive and one negative entry.
\end{proof}

We next establish an intuitive result: given two points $\pvec,
\qbase$, there is a budget $b$ such that the
constrained-optimal report is the midpoint between $\pvec$ and
$\qbase$.
\begin{lemma}\label{lemma:midpoint}
 Consider a scoring rule $\Svec$ that satisfies the budget-constrained truthful
property. For any pair $\pvec, \qbase \in \Delta$, there is a
$b$ such that $\qvec^*(\pvec, \qbase,b) = 0.5\qbase +
0.5\pvec$.
\end{lemma}
\begin{proof}
 Let $b = \nb(\qbase, 0.5\qbase+0.5\pvec)$. We must have $b>0$, because if $b \le
0$, then an agent who believed $\qbase$ would make at least as
much expected profit by reporting $0.5\qbase + 0.5\pvec$
instead of being truthful; this would violate the properness of
the scoring rule.

Given the definition of the budget-constrained truthful
property, we only need to show that, for any $\alpha > 0.5$,
$\nb(\qbase, \alpha\pvec + (1-\alpha)\qbase)>b$.  Without loss
of generality, let us assume that $S_X(\qbase) - S_X(0.5 \pvec
+ 0.5\qbase) = b>0$.  Given a $\alpha >0.5$, we note that
$(0.5\pvec + 0.5 \qbase)$ is a mixture of $\qbase$ and $(\alpha
\pvec + (1-\alpha)\qbase$.  By our quasiconcavity assumption,
it must be true that $S_X(\alpha \pvec + (1-\alpha) \qbase <
S_X(0.5\pvec + 0.5\qbase)$. But then, by the definition of the
budget, we must have $\nb(\qbase, \alpha\pvec +
(1-\alpha)\qbase) >b$.
\end{proof}

The following lemma will allow us to construct a contradiction:
\begin{lemma}\label{lemma:construction}
 Suppose that $\Svec$ is continuous, differentiable, and satisfies the
budget-constrained optimality property. Then, there are
probability distributions $\pvec, \qvec, \qbase, \pvec',
\qbase'$ and budgets $b,b'$ such that the following properties
are satisfied:
\begin{enumerate}
\item $\qvec = 0.5\pvec + 0.5 \qbase$ and $\qvec =
    \qvec^{*}(\pvec, \qbase, b)$. More specifically,
    $S_X(\qbase) - S_X(\qvec)=b$
and $S_Y(\qbase) - S_Y(\qvec) < b$, $S_Z(\qbase) -
S_Z(\qvec) < b$.
\item $\qvec = 0.5\pvec' + 0.5 \qbase'$ and $\qvec =
    \qvec^{*}(\pvec', \qbase', b')$. More specifically,
    $S_X(\qbase') - S_X(\qvec)=b'$
and $S_Y(\qbase') - S_Y(\qvec) < b'$, $S_Z(\qbase') -
S_Z(\qvec) < b'$.
\item $\frac{p_Y}{p_Z} \neq \frac{p'_Y}{p'_Z}$
\end{enumerate}
\end{lemma}
\begin{proof}
 We pick arbitrary $\pvec$ and $\qbase$ in the interior of the space of
possible distributions, and let $\qvec = 0.5\qbase + 0.5 \pvec$
be a mixture of these two distributions. Let $b = \nb(\qbase,
\qvec)$. By lemma~\ref{lemma:midpoint},  $\qvec =
\qvec^*(\pvec, \qbase, b)$.
 Now, the budget $b$ is the worst-case loss of moving from $\qbase$ to $\qvec$;
without loss of generality, assume that this loss occurs with
outcome $X$. Thus, we must have:
\[
   S_X(\qbase) - S_X(\qvec) = b
\]
 Suppose that we also have $S_Y(\qbase) - S_Y(\qvec) = b$. Then, it must be the case that $S_Z(\qbase) < S_Z(\qvec)$.
By Lemma~\ref{lemma:local}, we can perturb our choice of
$\pvec$ slightly such that only $S_X(\qbase) - S_X(\qvec) = b$,
while $S_Y(\qbase) - S_Y(\qvec) < b$ and $S_Z(\qbase) -
S_Z(\qbase) <b$.

Now, consider $\qbase' = \qbase - \eps$ and $\pvec' = \pvec' +
\eps$, where $\eps$ is a small perturbation that is constrained
to ensure that $\qbase'$ and $\pvec'$ are still probability
distributions. Note that $\qvec = 0.5\qbase' + 0.5\pvec'$. By
lemma~\ref{lemma:midpoint}, there is some budget $b'$ such that
$\qvec = \qvec^{*}(\pvec', \qbase', b')$. By the continuity of
$\Svec$, for all $\eps$ within a sufficiently small ball, it
must still be true that $S_X(\qbase) - S_X(\qvec)  >
S_Y(\qbase) - S_Y(\qvec), S_Z(\qbase) - S_Z(\qvec)$.

To get the last part of the construction, we pick $\eps$ such
that $p'_Y > p_Y, p'_Z < p_Z$.
\end{proof}

Finally, we can use the constructed structure of
Lemma~\ref{lemma:construction} to prove a contradiction.
\begin{theorem}\label{thm:impossibility}
  There is no scoring function $\Svec$ that satisfies the technical conditions
(continuity, differentiability, and quasiconcavity) and also
satisfies the budget-constrained truthfulness property.
\end{theorem}
\begin{proof}
Suppose that there is such a scoring function. Consider $\pvec,
\qvec, \qbase, \pvec', \qbase'$ and budgets $b,b'$ as
constructed in Lemma~\ref{lemma:construction}.

By assumption, if the initial position is $\qbase$, and an
agent has belief $\pvec$ and budget $b$, the optimal report is
$\qvec$. With this report, the agent's expected score is $\pvec
\cdot S(\qvec) - \pvec \cdot S(\qbase)$.

Now, consider the neighborhood of the report $\qvec$. By
Lemma~\ref{lemma:local}, there is a direction $\eps$ such that
$S_X(\qvec + \eps d\delta) = S_X(\qvec)$. By
lemma~\ref{lemma:construction}, and the continuity of $\Svec$,
this means that $\nb(\qbase, \qvec + \eps d\delta) =
\nb(\qbase, \qvec)=b$, as locally $S_X$ is the binding
constraint on the budget. In other words, the report $\qvec +
\eps d\delta$, as well as $\qvec - \eps d\delta$, are in the
feasible set for an agent with budget $b$.  As $S_X$ is
constant along direction $\eps$, the factors affecting the
agent's expected score are the change in $S_Y$ and $S_Z$. The
first-order conditions for optimality at $\qvec$ then give us:
\[
  p_Y \frac{dS_Y(\qvec+ \eps \delta)}{d\delta} +  p_Z \frac{dS_Z(\qvec+ \eps \delta)}{d\delta} = 0
\]
Rearranging, we must have:
\[
  \frac{\frac{dS_Z(\qvec+ \eps \delta)}{d\delta}}{\frac{dS_Y(\qvec+ \eps \delta)}{d\delta}} = - \frac{p_Y}{p_Z}
\]

Note that the left hand side depends only on $\qvec$.
Therefore, if we repeat our analysis with $\pvec'$, $\qbase'$,
and $b'$ we would derive that the same LHS  (with the same
vector $\eps$) is equal to $-\frac{p'_Y}{p'_Z}$. By the last
assertion in Lemma~\ref{lemma:construction}, this is a
contradiction.
\end{proof}

 Thus, for unrestricted distributions over 3 or more outcomes, it is impossible to
find a smooth scoring rule that will always incentivize agents
to move towards their true belief, given the natural budget
constraint.

In fact, we can prove an even stronger negative result about the aggregative properties of the market scoring rule in the presence of hard budget constraints: We can show that the optimal report of a budget-constrained agent can be locally insensitive to changes in her belief $\pvec$. This implies that future traders may not be able to infer this trader's belief from her trade, thus resulting in imperfect aggregation of information even if all budgets are known.

In order to prove this result, we first prove two technical lemmas. The first 
result just shows that there is a reports $\rvec$ such that two budget constraints are tight.
\begin{lemma}\label{lemma:xybudget}
 Given a strictly proper scoring rule $\Svec$ that satisfies the technical assumptions (differentiability and quasi-concavity), and given any starting probability distribution $\qbase$ in the interior of the space of distributions over $\{X,Y,Z\}$, there is a distribution $\rvec$ and constants $a,b >0$ such that two of the budget constraints with budget $b$ are tight:
 \[
   S_X(\rvec) - S_X(\qbase) = -b; S_Y(\rvec) - S_Y(\qbase) = -b; S_Z(\rvec) -S_Z(\qbase) = a
 \]
\end{lemma}
\begin{proof}
 Consider the set of all distributions as two-dimensional space. As $\qbase$ is an interior distribution, there is a circle $Q=\{\qvec\}$ of distributions such 
 that $|\qvec-\qbase|=c$. Consider the function $\nb(\qbase,\qvec)$ on $Q$. As 
 $Q$ is a compact space, and $\nb$ is continuous, this function must have a minimum value, call it $2b$, that is achieved within $Q$. As this is a proper scoring rule, we must have $2b>0$.
 
  Now, consider some distribution $\rvec'$ such that $\nb(\qbase,\rvec')=b$.
  Thus, $\rvec'$ must lie within the circle $Q$. Without loss of generality, let us suppose that $S_X(\rvec') - S_X(\qbase) = -b$. Then, by lemma~\ref{lemma:local}, there is a direction of movement around $\rvec$ such that $S_X(.)$ is constant, while $S_Y(.)$ decreases. Moving along this direction, we must reach a point $r$ at which both $S_X(\rvec) - S_X(\qbase)= -b$ and $S_Y(\rvec) - S_Y(\qbase) = -b$. At this point $\rvec$, we must have 
 $a \defeq S_Z(\rvec) - S_Z(\qbase) >0$, because otherwise $\qbase$ would earn a higher score than $\rvec$ for all three outcomes, which is impossible in a proper scoring rule.
\end{proof}
Next, we show that for beliefs in the neighborhood of such a point $\rvec$, we can find a 2-dimensional open set with certain properties:
\begin{lemma}\label{lemma:xyperturb}
 Given the construction of Lemma~\ref{lemma:xybudget}, consider perturbations to the distribution $\rvec$ parameterized by 
a pair $(\epsilon_X, \epsilon_Y)$. Define $\epsilon_Z = -(\epsilon_X + \epsilon_Y)$ and consider distributions $\pvec = \rvec+ \eps$. Then,
there is an open neighborhood $N$ of pairs $(\epsilon_X, \epsilon_Y)$ 
(or equivalently, a neighborhood of distributions $\pvec$)
such that the following conditions hold:
\begin{enumerate}
\item $\epsilon_X < 0, \epsilon_Y <0, \epsilon_Z >0$
\item $S_Z(\pvec) \geq S_Z(\rvec) - a$
\item Either $S_X(\pvec) < S_X(\rvec)$ or $S_Y(\pvec) < S_Y(\rvec)$, or both.
\end{enumerate}
\end{lemma}
\begin{proof}
 Consider the set $N_1$ of all $\eps$ such that condition (1) is satisfied, and $\eps$ is within a small enough ball that $\pvec$ is a valid distribution; this set is an open set. Now, consider all $\eps$ within a small ball such that 
 $S_Z(\pvec) \geq S_Z(\rvec) - a$. Because the scoring rule is continuous, this set must include an open set $N_2$ around $\eps=(0,0,0)$. Now, $N_1$ and $N_2$ have a non-empty intersection $N = N_1 \cap N_2$; $N$ is also an open set. 
 
 Thus, it remains to show that, for all $\eps$ in $N$, condition (3) is satisfied. We prove this by contradiction. Suppose that there was a $\eps \in N$ such that condition (3) was not satisfied, i.e., $S_X(\pvec)\geq S_X(\rvec)$ and $S_Y(\pvec) \geq S_Y(\rvec)$. We must have $S_Z(\pvec) < S_Z(\rvec)$. Then, we would have:
\begin{eqnarray*}
 (\pvec-\rvec).[\Svec(\pvec) - \Svec(\rvec)] &=& \eps.[\Svec(\pvec) - \Svec(\rvec)]  < 0 \\
 \Rightarrow \pvec.[\Svec(\pvec) - \Svec(\rvec)] &<& \rvec.[\Svec(\pvec) - \Svec(\rvec)]
\end{eqnarray*}
 The first inequality follows from the constructed signs of $\eps_X, \eps_Y$, and $\eps_Z$. However, for a proper scoring rule, the left hand side of the last inequality is positive, while the right hand side is negative, leading to a contradiction.
\end{proof}

Now, we can show that, for all beliefs $\pvec$ in this 2-dimensional open set $N$, the optimal report would be the same, $\rvec$.
\begin{theorem}\label{thm:insensitive}
 Given any proper scoring rule satisfying the technical assumptions, and any interior starting position $\qbase$, there is a 2-dimensional open set $N$ of 
 beliefs $\pvec$, a budget $b$, and a feasible report $\rvec$ such that, for all beliefs in $N$, the optimal budget-constrained report is $\rvec$. 
\end{theorem}
\begin{proof}
 Construct $N, b$, and $\rvec$ according to Lemmas~\ref{lemma:xybudget} and Lemma~\ref{lemma:xyperturb}. Consider $\pvec$ in $N$. We can see, by the third condition in Lemma~\ref{lemma:xyperturb}, that $\pvec$ is not a feasible report given budget $b$. 
 
  Now, consider any feasible report $\qvec \neq \rvec$. As $\qvec$ is feasible, we must have $S_X(\qvec) \geq S_X(\rvec)$ and $S_Y(\qvec) \geq S_Y(\rvec)$. 
  As $\Svec$ is a proper scoring rule, it follows that we must have $S_Z(\qvec) < S_Z(\rvec)$. Now, for $\pvec$ and $\eps = \pvec - \rvec$, we have:
  \begin{eqnarray*}
   (\pvec - \rvec).[\Svec(\qvec) - \Svec(\rvec)] &=& \eps.[\Svec(\qvec) - \Svec(\rvec)] <0 \\
  \Rightarrow \pvec.[\Svec(\qvec) - \Svec(\rvec)] &<& \rvec.[\Svec(\qvec) - \Svec(\rvec)]  
  \end{eqnarray*}
 As $\Svec$ is proper, the RHS of the last inequality is negative. Thus, 
 the LHS must also be negative, implying that, under belief $\pvec$, $\rvec$ gives a higher expected score than $\qvec$. As this is true for all feasible $\qvec$, $\rvec$ must be the optimal report.
\end{proof}

 This result shows that failure of information aggregation cannot be avoided with the market scoring rule, {\em even if all budgets are common knowledge}: For an open set of beliefs, the corresponding report is the same $\rvec$, and 
 hence future traders cannot infer or condition on the precise belief held by a trader. For forecasting problems with $k>3$ outcomes, we believe that it should be possible to prove an extension to theorem~\ref{thm:insensitive} that shows that the report is insensitive to beliefs in a $(k-1)$-dimensional neighborhood; this is a direction for future work.

\section{Impossibility results for multi-dimensional events}\label{sec:2dim}
For some multi-outcome events, we may have reasonable
constraints on the family of distributions that agents may
believe and report. Of course, there are innumerable different
ways of defining constrained distributions over the outcomes.
In this section, we consider one very natural form of
constraint that arises when the outcome itself can be
decomposed into independent dimensions. For example, when
eliciting beliefs about the likelihood of a hurricane striking
a city this year and next year, it may be natural to assume
that, although there are four possible outcomes overall, the
likelihood of a strike is independent in the two years, and as
such, we can elicit independent probabilities for each year. In
this section, however, we show that
Theorem~\ref{thm:impossibility} extends to this class of
restricted preferences as well.

Without loss of generality, we focus on a model with two
dimensions of outcome and two possibilities in each dimension.
(The impossibility result directly follows for richer outcome
and report spaces.) Suppose that the outcome consists of a
pair, with the first component either {\em Top} or {\em Bottom},
and the second component either {\em Left} or {\em Right}. Each
forecast $\qvec$ is also a pair $(q_T, q_L)$. In this model, we
assume that the events {\em Top} and {\em Left} are
independent, so that these two parameters suffice to determine
the outcome. We will use the notation $\qhat$ to denote the
vector $(q_T q_L, q_T(1-q_L), (1-q_T)q_L, (1-q_T)(1-q_L))$,
i.e., the probabilities of $(TL, TR, BL, BR)$ implied by a
report vector $\qvec = (q_T, q_L)$.

Given the set of pairwise outcomes $\{TL, TR, BL, BR\}$, we can
represent a scoring rule as a set of four functions:
\[
  \Svec(\qvec) = [S_{TL}(\qvec), S_{TR}(\qvec), S_{BL}(\qvec), S_{BR}(\qvec)]
\]
Note that we do {\em not} assume that the scoring rule is
separable into separate scoring rules for each dimension. An
agents' score may depend arbitrarily on the pair of outcomes.

Theorem~\ref{thm:impossibility} does not immediately imply an
impossibility in this domain, because not all distributions
over $\{TL, TR, BL, BR\}$ are expressible in this model. The
set of reports is two-dimensional, $\Delta = \Delta(\{T,B\})
\times \Delta(\{L, R\})$. (If we had not assumed independence
of {\em Top} and {Left} events, but instead collected a
separate probability for each of the four outcomes, then
Theorem~\ref{thm:impossibility} would immediately apply.)
 However, we will show below that a very similar
construction works in this case as well.

We first prove an analogue of Lemma~\ref{lemma:local}:
\begin{lemma}\label{lemma:2dlocal}
  Consider a point $\pvec$ in the interior of $\Delta$, and consider a
neighboring point of the form $\pvec + \delta \eps$, where
$\eps$ is a nonzero vector such that $(\epsilon_T,
\epsilon_L)$, and $\delta$ is an infinitesimally small
quantity.  Then, there exists a direction $\eps$ such that
$\frac{dS_{TL}(\pvec + \delta \eps)}{d\delta} =0$.  Further,
for this $\eps$, out of the three other directional derivatives
$\frac{dS_{TR}(\pvec + \delta \eps)}{d\delta}$,
$\frac{dS_{BL}(\pvec + \delta \eps)}{d\delta}$, and
$\frac{dS_{BR}(\pvec + \delta \eps)}{d\delta}$, at least one
must be strictly positive, and at least one must be strictly
negative.
\end{lemma}
\begin{proof}
 Expanding $S_X$ around $\pvec$ in terms of its partial derivatives, we have:
\[
  \frac{dS_X(\pvec+ \delta \eps)}{d\delta} = \epsilon_T \frac{\partial S_{TL}}{\partial p_T} +  \epsilon_L \frac{\partial S_{TL}}{\partial p_L}
\]
Setting the LHS to $0$, we get one linear constraint on
$\epsilon_T$ and $\epsilon_L$; any solution to this will meet
our purposes.

For the second part of the statement, consider the point $\qvec
= \pvec+ \eps d\delta$. We have $S_{TL}(\qvec) =
S_{TL}(\pvec)$. However, as this is a proper scoring rule, we
must have $\phat\cdot[\Svec(\pvec) -\Svec(\qvec)] >0$ and
$\qhat\cdot[\Svec(\pvec) -\Svec(\qvec)] <0$.  As $\phat$ and
$\qhat$ have non-negative entries, the only way in which this
can happen is if $\Svec(\pvec) - \Svec(\qvec)$ has at least one
strictly positive and at least one strictly negative entry.
\end{proof}

Lemma~\ref{lemma:midpoint} does apply in this setting.
Therefore, we move on to prove an analogue of
Lemma~\ref{lemma:construction}:

\begin{lemma}\label{lemma:2dconstruction}
 Suppose that $\Svec$ is continuous, differentiable, and satisfies the
budget-constrained optimality property. Then, there are
probability distributions $\pvec, \qbase, \qvec=0.5\pvec +
0.5\qbase$, a ball radius $r>0$ and a budget function
$b(\pert)$ defined for all $\pert: |\pert| < r$ such that the
following property is satisfied.
  Let $\qbase(\pert) = \qbase - \pert$ and $\pvec(\pert) = \pvec + \pert$. Then,
$  \forall \pert \mbox{ s.t. } |\pert| < r$,
\[
 \qvec = \qvec^{*}(\pvec(\pert), \qbase(\pert), b(\pert))
\]
and
\begin{eqnarray*}
S_{TL}(\qbase(\pert)) - S_{TL} (\qvec) = b(\pert) > \max \{&&
S_{TR}(\qbase(\pert)) - S_{TR} (\qvec),\\&&
S_{BL}(\qbase(\pert)) - S_{BL} (\qvec),
S_{BR}(\qbase(\pert)) - S_{BR} (\qvec) \}
\end{eqnarray*}
\end{lemma}
\begin{proof}
 We pick arbitrary $\qbase$ in the interior of the space of
possible distributions, and pick a $\qvec$ sufficiently close
to $\qbase$ such that $\qvec + (\qvec-\qbase)$ and $\qbase +
(\qbase-\qvec)$ are both valid probability distributions. Let
$\pvec = \qvec + (\qvec-\qbase)$, so that  $\qvec = 0.5\qbase +
0.5 \pvec$. Let $b = b(\qbase, \qvec)$. By
lemma~\ref{lemma:midpoint},  $\qvec = \qvec^*(\pvec, \qbase,
b)$. Now, the budget $b$ is the worst-case loss of moving from
$\qbase$ to $\qvec$; without loss of generality, assume that
this loss occurs with outcome $TL$. Thus, we must have:
\[
   S_{TL}(\qbase) - S_{TL}(\qvec) = b
\]

 Now, consider the other components of the score difference:
$(
   S_{TR}(\qbase) - S_{TR}(\qvec),
   S_{BL}(\qbase) - S_{BL}(\qvec),
   S_{BR}(\qbase) - S_{BR}(\qvec)
)$. We first show that we can perturb the choice of $\qvec$
(and perhaps permute the outcome names) such that all of these
components are strictly less than $b$.

Note that at least one of these components must be negative as
the scoring rule is proper when the true distribution is
$\qvec$. If exactly one of these components is negative, we can
swap $\qbase$ and $\qvec$ (changing $\pvec$ accordingly), such
that there is a unique positive component among all four score
differences, which would satisfy the requirement.

Thus, we can restrict our attention to the case in which one
component -- say $S_{TR}(\qbase) - S_{TR}(\qvec)$ is positive,
while the other two components are negative. In this case,
there is a problem only if $S_{TR}(\qbase) - S_{TR}(\qvec) =
b$. Now, let us apply Lemma~\ref{lemma:2dlocal}. This
guarantees a direction of change such that $S_{TL}$ is
unchanged. If $S_{TR}$ changes along this direction, then we
can slightly perturb $\qvec$ to get this property. If $S_{TR}$
is unchanged, then, by lemma~\ref{lemma:2dlocal}, both $S_{BL}$
and $S_{BR}$ must change, in opposite directions. In the latter
case, we can again swap $\qbase$ and $\qvec$; this would
guarantee that exactly two components are positive, and a small
perturbation would guarantee that one of them is uniquely the
maximum.

Now, we can safely assume that we have picked $\qbase$ and
$\qvec$ such that the only tight budget constraint is that
$S_{TL}(\qbase) - S_{TL}(\qvec) = b>0$. Consider $\qbase(\pert)
= \qbase - \pert$ and $\pvec(\pert) = \pvec' + \pert$, where
$\pert$ is a small perturbation. By the continuity of $\Svec$,
for all $\pert$ within a ball of sufficiently small radius $r$,
it must still be true that $S_{TL}(\qbase) - S_{TL}(\qvec)  >
S_{TR}(\qbase) - S_{TR}(\qvec), S_{BL}(\qbase) - S_{BL}(\qvec),
S_{BR}(\qbase) - S_{BR}(\qvec)$. Note that $\qvec =
0.5\qbase(\pert) + 0.5\pvec(\pert)$. By
lemma~\ref{lemma:midpoint}, there is some budget $b(\pert)$
such that $\qvec = \qvec^{*}(\pvec(\pert), \qbase(\pert),
b(\pert))$.
\end{proof}

Finally, we can use the constructed structure of
Lemma~\ref{lemma:2dconstruction} to prove a contradiction. This
proof is analogous to Theorem~\ref{thm:impossibility}

\begin{theorem}\label{thm:2dimpossibility}
  For the two-dimensional outcome setting, there is no scoring function $\Svec$ that satisfies the technical conditions (continuity, differentiability, and quasiconcavity) and also satisfies the budget-constrained truthfulness property.
\end{theorem}
\begin{proof}
Suppose that there is such a scoring function. Consider
$\pvec(\pert), \qvec, \qbase(\pert),b(\pert)$ as constructed in
Lemma~\ref{lemma:2dconstruction}.

Fix an arbitrary $\pert$ within the ball, and let us use the
shorthand $\pvec, \qbase, b$  for $\pvec(\pert), \qbase(\pert),
b(\pert)$ respectively. By assumption, if the initial position
is $\qbase$, and an agent has belief $\pvec$ and budget $b$,
the optimal report is $\qvec$. With this report, the agent's
expected score is $\phat \cdot S(\qvec) - \phat \cdot
S(\qbase)$.

Now, consider the neighborhood of the report $\qvec$. By
Lemma~\ref{lemma:2dlocal}, there is a direction $\eps$ such
that $S_{TL}(\qvec + \eps d\delta) = S_{TL}(\qvec)$. By
lemma~\ref{lemma:2dconstruction}, and the continuity of
$\Svec$, this means that $b(\qbase, \qvec + \eps d\delta) =
b(\qbase, \qvec)=b$, as locally $S_{TL}$ is the binding
constraint on the budget. In other words, the report $\qvec +
\eps d\delta$, as well as $\qvec - \eps d\delta$, are in the
feasible set for an agent with budget $b$.  As $S_{TL}$ is
constant along direction $\eps$, the factors affecting the
agent's expected score are the change in $S_{TR}, S_{BL}$ and
$S_{BR}$. The first-order conditions for optimality at $\qvec$
then give us:
\[
  \phat \cdot \left(
  \frac{dS_{TL}(\qvec+ \eps \delta)}{d\delta} +
  \frac{dS_{TR}(\qvec+ \eps \delta)}{d\delta} +
  \frac{dS_{BL}(\qvec+ \eps \delta)}{d\delta} +
  \frac{dS_{BR}(\qvec+ \eps \delta)}{d\delta}
 \right) = 0
\]

Note that $\qvec$ and $\eps$ are independent of our choice of
$\pert$. Thus, for all $\pert$ in an open neighborhood, we must
have:
\begin{equation}
 \phat(\pert) \cdot \Dvec = 0
\label{eq:subspace}
\end{equation}

where $\Dvec$ is a vector with first component $0$, and at
least one positive and one negative component. In other words,
all $\phat(\pert)$ corresponding to an open set of reported
probabilities $\pvec(\pert)$ must lie in a single linear
subspace of 4-dimensional space.

We can now demonstrate a contradiction, by showing that, as the
set of feasible $\phat$ is curved, no linear subspace can hold
an open neighborhood. Consider an arbitrary estimate $\pvec_0 =
(x,y)$ within this neighborhood, and three points in its
neighborhood $\pvec_1 = (x + \delta, y)$, $\pvec_2 = (x,
y+\delta)$, and $\pvec_3 = (x + \delta, y + \delta)$, where
$\delta>0$.

Equation~\ref{eq:subspace} must then hold for $\phat_0,
\phat_1, \phat_2, \phat_3$. Thus, we must get:
\begin{eqnarray*}
(\phat_1 - \phat_0)\cdot \Dvec = 0 & \Rightarrow & \delta (y, 1-y, -y, -1+y)\cdot \Dvec=0 \\
(\phat_2 - \phat_0)\cdot \Dvec = 0 & \Rightarrow & \delta (x, -x, 1-x, -1+x)\cdot \Dvec=0 \\
(\phat_3 - \phat_1)\cdot \Dvec = 0 & \Rightarrow & \delta (x+ \delta, -x-\delta, 1-x-\delta, -1+x+\delta)\cdot \Dvec=0 \\
  \mbox{(subtract second eq from third)}                                & \Rightarrow & (1,-1, -1, 1)\cdot\Dvec=0 \\
  \mbox{(first eq. } - \delta y \times \mbox{ fourth eq.)}              & \Rightarrow & (0,1, 0, -1)\cdot\Dvec=0 \\
  \mbox{(second eq. } - \delta x \times \mbox{ fourth eq.)}              & \Rightarrow & (0,0, 1, -1)\cdot\Dvec=0 \\
\end{eqnarray*}
Collecting all the equalities, we have:
\[
(1,0,0,0)\cdot\Dvec=0 ; \;\;
(1,-1,-1,1)\cdot\Dvec=0 ; \;\;
(0,1,0,-1)\cdot\Dvec=0 ; \;\;
(0,0,1,-1)\cdot\Dvec=0
\]
(The first equality is because, by construction, the first
component of $\Dvec$ is $0$.) However, the only solution to
this set of equations is $\Dvec = (0,0,0,0)$. This contradicts
the property, guaranteed by Lemma~\ref{lemma:2dlocal}, that
$\Dvec$ has at least two nonzero entries.
\end{proof}

\section{A mechanism for budget-constrained elicitation}\label{sec:mechanism}
In this section, we describe a mechanism, the Scaled Scoring
Mechanism, for sequential information aggregation that modifies
the market scoring rule to address the limitations under budget
constraints. The mechanism and its characterization are fairly
simple, but it serves to illustrate how designing for budget
constraints can be useful. The sequential update and scoring
rules we use were previously proposed for binary outcomes in
the context of recommender systems by Resnick and
Sami~\cite{RS07}; that paper did not address the properties for
multi-outcome events or the possibility of mis-reporting
budgets that we focus on here.

We make one important assumption: Our mechanism will ask users
to report their current budget $b$, but {\em we assume that it
is impossible for an agent to report a higher budget than her
true budget $b$}, although users may strategically under-report
$b$. Such a mechanism could be sustained by requiring each
agent to deposit her entire reported budget $b$ up front, at
the time of making her report. If the budget is truly a hard
constraint, then an agent will not be able to deposit more than
$b$, even if she is sure that she cannot lose all that she
deposits.

Consider a sequence of agents $1, 2, \ldots m$ interacting with
a mechanism, in which the goal of the mechanism designer is to
elicit an unrestricted distribution $\qvec$ over a given
outcome space. Let $\Svec$ be a proper scoring rule that
satisfies the following properties: Each component of
$\Svec(\qvec)$ is concave in $\qvec$, and for all allowed
$\qvec, \qbase$, $\nb(\qbase, \qvec) \leq B$. In other words,
the maximum budget required for any feasible move is bounded by
$B$. For example, $\Svec$ may be the quadratic scoring rule. If
the range of allowed distributions is restricted slightly so
that all individual probabilities are bounded away from $0$, we
could even use the logarithmic scoring rule $\Svec(\qvec) =
(\log q_x, \log q_y, \log q_z)$.

The {\em Scaled Scoring Mechanism} operates as follows:
\begin{enumerate}
\item At any time, the mechanism has a forecast $\qbase$
    that also serves as a reference point. Initialize
    $\qbase_0=q_0$, an arbitrary initial default
    prediction.
\item For $i=1,2, \ldots, m$:
  \begin{enumerate}
    \item Ask agent $i$ to report a budget (denoted by
        $b'_i$) and a forecast $\qvec_i$.
    \item Update
         \[
            \qbase_i = \qbase_{i-1} + \max \left\{ 1, \frac{b'_i}{B}\right\} (\qvec_i - \qbase_{i-i})
         \]
  \end{enumerate}
\item When the outcome is revealed (say, to be $X$), score
    the agents as follows: each agent $i$ gets a net payoff
    equal to:
\[
  \max \left\{ 1, \frac{b'_i}{B} \right\} \left[S_X(\qvec_i) - S_X(\qbase_{i-1})\right]
\]
\end{enumerate}

We make the following observations about the scaled scoring
mechanism: Firstly, when reported budgets are unlimited ($b'_i
> B$), the updates and scoring are exactly the same as the
market scoring rule. However, when budgets are limited, the
score is a scaled-down value proportional to the market scoring
rule score, {\em and} the reference point is moved only a
fraction of the distance to the reported belief $\qvec_i$.

 We can prove the following properties about the Scaled Scoring Mechanism:
\begin{theorem}\label{theorem:SSM}
 The Scaled Scoring Mechanism satisfies the following properties:
\begin{enumerate}
\item \label{prop:honesty} Myopic strategyproofness:
    Assuming that each agent $i$ cannot report $b'_i$
    higher than her true
budget $b_i$, and that each trader trades just once, it is
optimal for each trader to report her true belief $\pvec_i$
and her true budget $b_i$.

\item \label{prop:budget} No default: Under any outcome,
    each agent $i$'s net payoff is at least $-b'_i$.
\item \label{prop:constrainedtruthful} Budget-constrained
    truthful: After $i$ selects her optimal report, the
    forecast
distribution $\qbase_i$ is a mixture of the earlier
forecast distribution $\qbase_{i-1}$ and $i$'s true belief
$\pvec_i$, for any budget $b_i$.

\item \label{prop:limitedloss} Limited Loss: The worst-case
    loss of the market maker is no worse than the
    worst-case
loss of a market scoring rule market for the same outcome
space, with the same initial distribution $\qbase_0$.

\item \label{prop:sybilproof} Myopic sybilproofness: It is
    never profitable for an agent to divide her budget
    among
multiple identities who make consecutive reports.%
\footnote{This is a relaxation of the path-invariance
property characteristic of market scoring rules, in which
the sum of payoffs of consecutive trades is equal to the
payoff of a single trade with the same final report.}
\end{enumerate}
\end{theorem}
\begin{proof}
(\ref{prop:honesty}): First, we observe that for any fixed
reported budget $b'_i$, the optimal report $\qvec_i$ is equal
to agent $i$'s belief $\pvec_i$: Her payoff is proportional to
$\Svec(\qvec_i) - \Svec(\qbase_{i-1})$, and as $\Svec$ is a
proper scoring rule, her expected payoff is maximized when
$\qvec_i = \pvec_i$. Now, with $\qvec_i=\pvec_i$, we note that
her expected payoff is:
\[
  \max \left\{ 1, \frac{b'_i}{B} \right\} \pvec_i . \left[S(\pvec_i) - S_X(\qbase_{i-1})\right]
\]
As $\Svec$ is a proper scoring rule, $\pvec_i.\left[S(\pvec_i)
- S_X(\qbase_{i-1})\right] $ is non-negative. Hence, her
expected payoff is non-decreasing in $b'_i$, and as she cannot
exaggerate her budget, it is optimal for her to also report her
true budget $b'_i = b_i$.

(\ref{prop:budget}):
  By definition of $B$, agent $i$'s payoff cannot be less than $- b'_i$.

(\ref{prop:constrainedtruthful}):
 From part~(\ref{prop:honesty}), the optimal strategy for $i$ is to truthfully
report her true belief $\pvec_i$ and her true budget $b_i$. As
the update to $\qbase$ is done by mixing with $i$'s report, the
budget-constrained truthfulness property holds.

(\ref{prop:limitedloss}):
 Here, we use the fact that each component of the scoring rule $\Svec$ is concave.
Let $\lambda = \max{1, \frac{b'_i}{B}}$. By concavity we have,
for outcome $X$:
\begin{eqnarray}
S_X(\qbase_i) &=& S_X((1-\lambda)\qbase_{i-1} + \lambda \qvec_i)\\
& \geq & (1-\lambda)S_X(\qbase_{i-1}) + \lambda S_X(\qvec_i)\\
\Rightarrow
S_X(\qbase_i) - S_X(\qbase_{i-1}) & \geq & \lambda [S_X(\qvec_i) - S_X(\qbase_{i-1})
\label{eq:concave}
\end{eqnarray}
Similar inequalities hold for other outcomes $Y, Z,..$. Note
that the right hand side of Equation~\ref{eq:concave} is
precisely the payment made to agent $i$ by the SSM mechanism.
The right hand side of the equation can be interpreted as the
payment that the MSR mechanism would have made to agent $i$ if
the sequence of reports was $\qbase_1, \qbase_2, \ldots,
\qbase_M$. Thus, the worst-case loss of the SSM mechanism, over
all budgets, report sequences, and outcomes, is no more than
the worst-case loss of the corresponding MSR mechanism.

(\ref{prop:sybilproof}):
 This property also follows from equation~\ref{eq:concave}. Suppose the
initial market prediction is $\qbase_{i-1}$. Consider a set of
consecutive reported distributions $\qvec_{i1}, \qvec_{i2},
\ldots, \qvec_{ij}$, with corresponding budgets $b_{i1},
b_{i2}, \ldots, b_{ij}$. With this sequence, the final
reference value under SSM would be some distribution
$\qbase_{ij}$. By equation~\ref{eq:concave}, the total payoff
for these $j$ reports would be less than the MSR payoff of a
single move to $\qbase_{ij}$.

We will argue that an agent $i$ with the combined budget of all
$j$ agents could make at least as high a profit with a single
report $\qvec_i$. Let $b = \sum_j b_{ij}$. If $b \geq B$, then
agent $i$ could always move to $\qbase_{ij}$, thereby earning
(under SSM) an amount equal to the MSR payoff of moving to
$\qbase_{ij}$.

If $b < B$, then consider an agent who believes a distribution
$\pvec_i$.  It suffices to consider the case $j=2$; higher
values of $j$ will follow by induction, because we can first
replace the last two trades with a single trade without hurting
profits, and then repeat this procedure. Let $u_0, u_1, u_1$
denote the expected scores of $\qbase_{i-1},\qvec_{i1},
\qvec_{i2}$ under this belief, and let $v_1$ denote the
expected score of $\qbase_{i1}$ under this belief.
 For conciseness, let us
define $\lambda_{t} \defeq \frac{b_{it}}{B}$ and $\lambar_{t}
\defeq 1 -\lambda_{t}$. By equation~\ref{eq:concave}, we see
that $v_1 \geq \lambar_1 u_0 + \lambda_1 u_1$. The total payoff
of the $2$ trades is seen to be:
\begin{eqnarray*}
\lambda_1 [u_1 - u_0] + \lambda_2 [u_2 - v_1] \\
&\leq& \lambda_1 [u_1 - u_0] + \lambda_2 [u_2 - \lambda_1 u_1 - \lambar_1 u_0] \\
&=& \lambda_1 [u_1 - u_0] + \lambda_2 [u_2 - u_0] - \lambda_2\lambda_1 [u_1-u_0] \\
&=& \lambda_1\lambar_2 [u_1-u_0] + \lambda_2 [u_2 - u_0] \\
&\leq& (\lambda_1 + \lambda_2) \max\{u_1-u_0, u_2 - u_0\}
\end{eqnarray*}
The right hand side of the last equation is the payoff of a
single trade that reported the better of $\qvec_{i1},
\qvec_{i2}$ (under the believed distribution) with budget $b =
b_{i1} + b_{i2}$. Thus, the single trade never has worse payoff
than the two consecutive trades. Inductively replacing two
consecutive trades with a single trade, we can show that this
result holds for $j>2$ as well.

\end{proof}

We note that both the strategyproofness and resistance to sybil
attacks is of a myopic nature, in that it may not hold if
agents can make trades before and after honest traders' trades.
However, this restriction seems unavoidable, as even in the
setting without budget constraints, such non-myopic attacks can
be profitable under some models of trader
beliefs~\cite{Chen-algorithmica}.

When we receive a forecast of $\qvec_i$ from agent $i$, and know that it is rational for agent $i$ to report this honestly, it may be tempting to
update the market forecast to $\qvec_i$ instead of
$\qbase_i$. However, if we used $\qvec_i$ as the reference forecast for the
next trade, the market would not satisfy the myopic sybilproofness property described above: A trader may profit from splitting her trade into two successive trades.

Note that, if we had updated the market forecast and reference probability

For each agent $i$, if the budget $b_i >0$ is known, $i$'s true
belief can be determined from the updated market forecast
$\qbase$. (Depending on the form of $i$'s signal and the
initial $\qbase_{i-1}$, this may even be possible if the
budgets are not known). Thus, in a model where the budgets are
common knowledge, each trader can condition on all past
traders' beliefs, as in the budget-unconstrained market. In
other words, the SSM market generically enables perfect
information aggregation.

The SSM mechanism gets around the impossibility result of
section~\ref{sec:3outcome} because it enforces a stronger
restriction on how an agent may use her budget. In particular,
it enforces a form of the {\em Kelly gambling}~\cite{Kelly}
strategy, in which an agent invests a fraction of her budget
proportional to the magnitude of her unconstrained move. The
Kelly gambling strategy is asymptotically optimal for long-term
budget growth, but it is suboptimal for expected short-term
budget growth. In our model, this manifests itself in the
following way: An agent is never able to stake her entire
budget in her trade, and consequently, the updated prediction
$\qbase_1$ is less informative than the updated prediction of
the MSR with the same budgets. However, this may be more than
offset by the improved ability for future traders to aggregate
information, and because the enforced proportional-betting
policy leads to better long-term budget growth where agents may
be acting myopically.

\section{Discussion}\label{sec:discussion}
\paragraph{Extension to Continuous Outcomes}
For geographical events, such as forecasting the location of a tornado strike, it may be natural to model the outcome space as continuous, and set the goal of forecasting the expected value. Lambert \etal~\cite{Lambert-elicitation} have shown that (without budget constraints) the mean can be elicited by a proper scoring rule. However, Theorem~\ref{thm:2dimpossibility} extends {\em a fortiori} to a model in which we want to forecast the mean of a 2-dimensional
outcome, with independence between the dimensions. A special
case of the continuous model is the case in which all beliefs
are over the outcomes $(0,0), (0,1), (1,0), (1,1)$;
theorem~\ref{thm:2dimpossibility} shows that no smooth scoring
rule can be budget-constrained truthful in this case, and hence
in general.

\paragraph{Separable scoring rules}
 For forecasting the mean of a multi-dimensional outcome, one natural approach is to elicit separate mean forecasts along each dimension, and to pay off each dimension's forecast through a separate scoring rule. In other words, we might have a mean latitude market and a mean longitude market. With such a setup, it does not matter whether an agent believes the two dimensions are independent or dependent; her expected payoffs would depend only on her marginal beliefs over the two dimensions. Theorem~\ref{thm:2dimpossibility} extends {em a fortiori} to this case as well: For a separable scoring rule, an agent's expected payoff is the same as another agent with the same marginal beliefs who believes the two dimensions are independent, and hence, her behavior will be the same.

\paragraph{Transformed independence}
Because we make no additional assumptions about the structure
of the scoring rule in section~\ref{sec:2dim}, the results
extend to many families of restricted distributions. In
particular, any family of restricted beliefs that can be
smoothly reparameterized to be independent along two
parameters will yield the same result.

\paragraph{Mechanisms for risk-averse agents}
We have modeled traders who are risk-neutral up to a hard budget constraint; this can be viewed as a specific family of risk-averse utility functions.
Dimitrov et al.~\cite{DSE09} showed that there are no attractive
sequential market mechanisms to aggregate information for {\em all}
unknown risk types: any such mechanism must exhibit exponentially reducing payoffs in order to guarantee myopic sybilproofness with a bounded subsidy. The scaled scoring mechanism does not have this negative property. Thus, we have
shows one natural class of restricted risk
preferences for which it is possible to sequentially aggregate information.
We note, however, that this relies on the assumption that you cannot exaggerate your budget; hence, we have access to a mechanism that reveals some information about a trader's risk type.

\paragraph{Limitations and directions for future work}
 Our impossibility results required some technical assumptions on the smoothness of the scoring rules. Intuitively, there does not seem to be any essential advantage to using non-smooth scoring rules, and one direction for future work is to relax these assumptions. It would also be helpful to extend theorem~\ref{thm:insensitive} to higher-dimensional neighborhoods.
 
 The Scaled Scoring Mechanism relies on the fact that a trader cannot deposit an amount greater than her true budget. However, depending on her reported belief, there may be no outcome in which she loses her entire deposit. Thus, if there is an external credit market from which she can borrow with her trade as collateral, this assumption may be violated. One important question for future work is to find alternative market-like mechanisms for sequential information aggregation in the presence of budget limits.

\section*{Acknowledgments}

This work was supported by the National Science Foundation under research grant IIS-0812042 and CCF-0829754.
The authors would like to thank Yiling Chen, Nicolas Lambert, John Langford, David Pennock and Daniel Reeves for helpful discussions. Thanks especially to David Pennock and his research group at Yahoo!\ Research in New York that provided the environment that fostered the research in this paper.

\bibliographystyle{alpha}
\bibliography{scoring-properties}
\end{document}